\documentclass[a4paper]{iopart}
\expandafter\let\csname equation*\endcsname\relax
\expandafter\let\csname endequation*\endcsname\relax
\usepackage{amsmath,amsfonts,amsthm,amssymb}
\usepackage{mathrsfs}
\usepackage[usenames,dvipsnames]{color}
\usepackage{calc}
\usepackage[colorlinks=true,linkcolor=blue,citecolor=red,urlcolor=Violet,bookmarksdepth=subsection,final]{hyperref}

\newtheorem{thm}{Theorem}[section]
\newtheorem{prop}[thm]{Proposition}
\newtheorem{lem}[thm]{Lemma}

\theoremstyle{definition}

\theoremstyle{remark}

  \newcommand{\oo}{\infty}
  \newcommand{\del}{\partial}

  \newcommand{\sse}{\subseteq}
  \newcommand{\sso}{\subset}
  \newcommand{\sm}{\setminus}
  \newcommand{\C}{\mathscr{C}}
  
\renewcommand{\d}{\mathrm{d}}
  \newcommand{\Diff}{\mathrm{Diff}}
  \newcommand{\eps}{\varepsilon}
  
  \newcommand{\G}{\mathscr{G}}
  \newcommand{\h}{\mathsf{h}}
  
  \newcommand{\Lie}{\mathcal{L}}
  \newcommand{\M}{\mathcal{M}}
  \newcommand{\rM}{\mathring{\mathcal{M}}}
  \newcommand{\cP}{\mathscr{P}}
  \newcommand{\bP}{\bar{\mathscr{P}}}
  \newcommand{\Secs}{\Gamma}
  \newcommand{\supp}{\operatorname{supp}}
  \newcommand{\U}{\mathscr{U}}
  
\begin{document}

\title{Local and gauge invariant observables in gravity}

\author{Igor Khavkine}
\address{Department of Mathematics, University of Trento,
and TIFPA-INFN, Trento,
I--38123 Povo (TN) Italy}
\ead{igor.khavkine@unitn.it}

\begin{abstract}
It is well known that general relativity (GR) does not possess any
non-trivial local (in a precise standard sense) and diffeomorphism
invariant observables. We propose a generalized notion of local
observables, which retain the most important properties that follow from
the standard definition of locality, yet is flexible enough to admit a
large class of diffeomorphism invariant observables in GR. The
generalization comes at a small price, that the domain of definition of
a generalized local observable may not cover the entire phase space of
GR and two such observables may have distinct domains. However, the
subset of metrics on which generalized local observables can be defined
is in a sense generic (its open interior is non-empty in the Whitney
strong topology).  Moreover, generalized local gauge invariant
observables are sufficient to separate diffeomorphism orbits on this
admissible subset of the phase space.  Connecting the construction with
the notion of differential invariants, gives a general scheme for
defining generalized local gauge invariant observables in arbitrary
gauge theories, which happens to agree with well-known results for
Maxwell and Yang-Mills theories.
\end{abstract}
%\keywords{...}
%\submitto{\CQG}
\maketitle

\section{Introduction}\label{sec:intro}
The goal of this note is to outline a connection between the theory of
differential invariants and local observables in gauge theories, in the
sense of classical and quantum field theory. The main example we will
treat is gravity, or more precisely general relativity (GR) possibly
coupled to matter fields, which is a gauge theory with diffeomorphisms as
the group of gauge transformations. The differential invariants in this
case are essentially scalars that can be tensorially constructed from
the Riemann curvature tensor and its covariant derivatives. The core
idea of the connection to local observables appeared already in the
proposal of Bergmann and Komar~\cite{bg0,bg1}. However, it seems, that
the idea has never been taken to the logical conclusion that we intend
to sketch below.

Consider the theory of a, say scalar, field $\phi$ on an $n$-dimensional
spacetime manifold $M$. The prototypical example of a local observable
in this theory is a smeared field
\begin{equation}
	\phi(f) = \int_M \phi(x) f(x) ,
\end{equation}
where the smearing test function $f\in \Omega^n(M)$ is $C^\oo$ with
compact support. Those last two properties are key to making $\phi(f)$ a
useful observable. Classically, an observable $F\colon \Phi \mapsto
F(\Phi)$ is a map from field configurations to real numbers. A smeared
field acts as $\phi(f)\colon \Phi \mapsto \int_M \Phi(x) f(x)$. The
compactness of the support of $f$ makes sure that this integral
converges for an arbitrary field configuration, so that $\phi(f)$ has a
large domain of definition on the phase space of the theory (on all of
it, in this case). The smoothness of $f$ makes sure that the Poisson
bracket
\begin{equation}\label{eq:loc2pois}
	\{\phi(f),\phi(g)\} = \int_{M\times M} f(x) E(x,y) g(y) ,
\end{equation} where
$E(x,y)$ is the distributional kernel of the Peierls formula and
$\phi(g)$ is a similar smeared field, is well defined as a
distributional integral. Compact support also helps with the convergence of
the Poisson bracket integral. Quantum mechanically, the field $\phi(x)$
is promoted to an operator valued distribution. The smoothness of the
smearing function $f$ is then essential to get an honest (though
unbounded) operator corresponding to $\phi(f)$. The expectation values
of products of smeared fields like
\begin{equation}
	\langle \phi(f) \phi(g) \rangle
	= \int_{M\times M} \langle \phi(x) \phi(y) \rangle f(x) g(y) ,
\end{equation}
are also distributional integrals with respect to the $2$-point singular
kernel $\langle \phi(x) \phi(y) \rangle$. Thus, the smoothness of $f$
and $g$ are again necessary to make sure that this integral is locally
well-defined (UV finite), with their compact support ensuring its global
convergence (IR finiteness). In short, we say that the smoothness of
test functions, like $f$, \emph{diffuses the UV singularities} of local
fields, like $\phi(x)$, and their compact support \emph{IR regularizes
them}.

An immediate generalization is the notion of a \emph{multilocal
observable}, which is given by a formula of the form
\begin{equation}
	\int_{M^m} \phi(x_1) \cdots \phi(x_m) f(x_1,\ldots,x_m) ,
\end{equation}
where the smearing test function $f\in \Omega^{mn}(M^m)$ is $C^\oo$ with
compact support. It should be noted that the Poisson bracket of two
local observables, as defined by Equation~\eqref{eq:loc2pois}, is in
general no longer a local observable. Rather, as in the example of
$\phi^2(f) = \int_M \phi^2(x)f(x)$, it is (almost) bilocal
(multilocal with $l=2$),
\begin{equation}
	\{\phi^2(f),\phi^2(g)\} =
	\int_{M\times M} 2\phi(x)f(x) \, E(x,y) \, 2\phi(y)g(y) ,
\end{equation}
with the caveat that the smearing function $f(x)E(x,y)g(y)$ is a
distribution and could be non-smooth. Thus, another natural generalization
that invites itself is that of multilocal observables with
distributional smearing, though the identification of the class of
distributions that can be consistently allowed becomes rather technical.
We mention these generalizations only for completeness, with the
remainder of this note concentrating on local observables with smooth
smearings. Though, we do briefly come back to multilocal observables in
Sections~\ref{sec:pois} and~\ref{sec:discuss}.

In the case when $\phi(x)$ is a local field in a gauge theory, another
important property demanded of a local observable like $\phi(f)$ is
\emph{gauge invariance}. That is, the value of $\phi(f)$ (numerical
value classically, and operatorial value quantum mechanically) stays
invariant under the action of gauge transformations. Any physically
meaningful quantity may only be represented by a gauge invariant
observable. It is common knowledge that, in gravitational theories, the
set of local gauge invariant observables is trivial (see for
instance~\cite{gmh} or \cite{bfr-qg}, for a clear discussion). Such a
statement can of course be made once a suitably precise notion of
locality and gauge invariance are given, as we do in
Section~\ref{sec:locobsv}. On the other hand, a slight relaxation of
that standard notion of locality, which we propose in
Section~\ref{sec:glocobsv}, opens the door to the introduction in
Section~\ref{sec:diffinv} of a large class of gravitational observables
that are gauge invariant (thanks to the use of differential invariants),
diffuse UV singularities and are IR regularizing. Finally, we address
the computation of Poisson brackets between generalized local gauge
invariant gravitational observables in Section~\ref{sec:pois}.
Ultimately, we propose to treat this generalized notion as the true
definition of local observables.

In the rest of the note we discuss only classical observables. Comments
on how the constructions outlined below impact perturbative quantum
field theory are left for the Discussion in Section~\ref{sec:discuss},
where we also mention various limitations and open problems of our
proposal.

We finish this section with a brief historical remark. The idea of
constructing observables in gravitational theories based on differential
invariants (curvature scalars) first appeared clearly in the works of
Bergmann and Komar~\cite{bg0,bg1}. Unfortunately, they never published
a computation of Poisson brackets for such observables. Such
computations appeared first in the work of DeWitt~\cite{dewitt}, who
used the Peierls bracket formalism. Since then, related ideas have appeared
sporadically in the literature, more recently referred to as relational
observables~\cite{tamb}. Some ideas in spirit similar to those presented
below can also be found in~\cite{gmh} and~\cite{bfr-qg}, with the latter
following-up a slightly different line of ideas that attempted to expand
the notion of local obsrvables by modifying the notion of gauge
invariance~\cite{rejzner-thesis,fr}.

\section{Standard local observables in field theories}\label{sec:locobsv}

Let us briefly set up the geometric formalism of classical field theory.
We will mostly follow the references~\cite{kh-caus,kh-peierls},
with~\cite{bf-notes,fr,hs,bfr} being complementary sources. We take $M$
to be an oriented $n$-dimensional smooth manifold.  Usually one endows
$M$ with a Lorentzian metric, but we are working at a level of
generality where that is not necessary. Take a vector bundle $F\to M$,
the \emph{field bundle}, and denote its sections as $\Phi \colon M\to
F$, a \emph{field configuration}. In more generality, $F\to M$ could be
a more general smooth bundle, but we will stick to the vector bundle
case for simplicity.

By $\pi^k\colon J^k F\to M$, for $k=0,1,\ldots,\oo$, we denote the
bundle of $k$-jets of the field bundle $F\to M$. Jets%
	\footnote{Jets are a standard constructions in differential geometry.
	An introduction to jets, operations on them and their applications to
	differential equations can be found in~\cite{olver}. See also the
	relevant appendices to~\cite{kh-caus,kh-peierls}.} %
naturally and geometrically capture information about higher derivatives
of sections of $F\to M$ over a point of $M$. Given a $k$-jet, throwing
away all the information about order-$k$ derivatives gives a
$(k-1)$-jet. In other words, we have natural bundle projections
$\pi^k_{k-1} \colon J^k F \to J^{k-1}F$ over $M$, until we get $J^0F =
F$. Any section $\Phi\colon M\to F$ can be naturally augmented with the
information about its derivatives (its jet) at every point of $M$, thus
defining the \emph{$k$-jet extension} section $j^k \Phi \colon M \to J^k
F$. To be more concrete, consider a fiber-adapted local coordinate
system $(x^i,\phi^a)$ on $F$. It induces an adapted local coordinate
system $(x^i,\phi^a_I)$ on $J^kF$ over that on $F$, where
$I=\varnothing,i,ij,\ldots$ ranges all possible multi-indices. The
coordinate system is adapted in the sense that the following identity
holds for any field section $\Phi$:
\begin{equation}
	\phi^a_{i_1\cdots i_l}(j^k\Phi(x))
	= \del_{i_1}\cdots \del_{i_l} \phi^a(\Phi(x)) .
\end{equation}

Next, we introduce the \emph{field configuration space} $\C = \Secs(F)$,
consisting of smooth sections of the vector bundle $F\to M$. It is an
infinite dimensional vector space. It is convenient to endow it with the
Whitney weak topology, which gives it the structure of a Fr\'echet
space~\cite{hirsch,km}. Unfortunately the Whitney weak topology is too
coarse for some of our purposes (its fundamental neighborhoods do not
control the behavior of sections toward the open ends of non-compact
manifolds), so we will mostly make use of the Whitney strong topology
(see the discussion in Section~\ref{sec:glocobsv}). Further, the
equations of motion of the field theory (e.g.,\ Klein-Gordon equation
for a scalar field, or Einstein's equations for the gravitational field)
select the subspace of solutions, $\cP \sso \C$, which we refer to as
the \emph{(covariant) phase space}. For non-linear equations, $\cP$ is
in general not a linear subspace of $\C$, however we will presume that
$\cP$ has a well-defined Fr\'echet manifold structure induced by its
inclusion as a submanifold of the Fr\'echet space $\C$. We are
ultimately interested in the \emph{algebra of observables} $C^\oo(\cP)$.
However, it is often more convenient to discuss elements of $C^\oo(\cP)$
as images of elements of $C^\oo(\C)$ under the projection induced by the
inclusion $\cP\sso \C$.  We make the simplifying assumption that this
inclusion is sufficiently regular for the projection to be surjective.
Then, strictly speaking, observables correspond to equivalence classes
of elements of $C^\oo(\C)$. However, we will not need to make use of
this distinction below and may also refer to elements of $C^\oo(\C)$ as
\emph{observables}, or alternatively as \emph{functionals}.

On $\C$, we can define a special class of functions called \emph{local
functionals (or observables)} with the help of \emph{horizontal forms}
on $J^kF$.  Horizontal forms, whose space we denote as
$\Omega^{p,0}(F,k) \sso \Omega^p(J^kF)$, are generated as linear
combinations from the pullback $(\pi^k)^* \Omega^p(M)$ of forms on the
spacetime with coefficients from $C^\oo(J^kF)$, meaning they are of the
form $\alpha_{i_1\cdots i_k}(x^i,\phi^a_I)\, \d{x}^{i_1} \cdots
\d{x}^{i_k}$. Of course, elements of $\Omega^{p,0}(F,k)$ can be pulled
back to $\Omega^{p,0}(F,l)$ along the natural jet projections $J^lF \to
J^kF$ for any $l>k$. It is convenient to take the \emph{increasing
union} (or \emph{direct limit}) $\Omega^{p,0}(F) = \bigoplus_{k=0}^\oo
\Omega^{p,0}(F,k)/{\sim}$, where the equivalence relation identifies a
form in $\Omega^{p,0}(F,k)$ with its pullback to any higher jet bundle,
so that we do not need to worry about the order $k$ when it is not
necessary. We call elements of $\Omega^{n,0}(F)$ \emph{horizontal
densities}. For any form $\alpha \in \Omega^p(J^kF)$, we define its
\emph{spacetime support} as the closure of the projection of its support
onto $M$, $\supp_M \alpha = \overline{\pi^k \supp \alpha}$.

It is helpful to note that any form $\beta \in \Omega^p(J^kF)$ can be
projected to a horizontal form $\h[\beta] = \alpha \in
\Omega^{p,0}(F,k+1)$, where the map acts as $\h[\d x^i] = \d x^i$ and
$\h[\d\phi^a_I] = \phi^a_{Ii} \d{x}^i$ on coordinate forms, extends
linearly and respects the wedge product. Another convenient operator to
define is the Euler-Lagrange derivative $\delta_{EL}$ of a horizontal
density $\alpha\in \Omega^{n,0}(F)$. Locally, we define
$\delta_{EL}[\alpha]$ by the following identity on $M$:
\begin{equation}\label{eq:EL-def}
	\left.\frac{\d}{\d t}\right|_{t=0} (j^k(\Phi + t\Psi))^* \alpha(x)
	= (j^k\Phi)^*\delta_{EL}[\alpha]_a(x) \Psi^a(x) + \d\xi[\Phi;\Psi] ,
\end{equation}
where each $\delta_{EL}[\alpha]_a \in \Omega^{n,0}(F)$ and $\xi$ is some
differential operator that depends linearly on its second argument.
Globally, $\delta_{EL}[\alpha]$ is a horizontal density valued in the
dual bundle $F^*\to M$. By the usual methods of variational calculus,
this relation makes $\delta_{EL}[\alpha]$ unique and well-defined. All
of these constructions, and more, naturally live in the context of the
\emph{variational bicomplex}~\cite{olver}, of which we shall not need to
make further use in this note.

To any horizontal density $\alpha \in \Omega^{n,0}(F)$ with compact
spacetime support, we can associate a functional
\begin{equation}\label{eq:locdef}
	A[\Phi] = \int_M (j^k\Phi)^* \alpha .
\end{equation}
If, in local adapted coordinates, we have $\alpha = \tilde{\alpha}(x^i,
\phi^a, \phi^a_i, \phi^a_{ij}, \ldots)\, \d^n x$, then
\begin{equation}
	A[\Phi] = \int_M \tilde{\alpha}(x^i,\phi^a(\Phi(x)), \del_i \phi^a(\Phi(x)),
		\del_i\del_j \phi^a(\Phi(x)), \ldots)\, \d^n x .
\end{equation}
It is straightforward to verify that, by the compact spacetime support
condition, the above integral converges for an arbitrary field
configuration $\Phi \in \C$ and in fact $A\in C(\C)$. Of course, we
would like $A$ to be not only continuous, but also in some sense smooth
on the infinite dimensional manifold $\C$. It is in fact possible to
make use of an infinite dimensional calculus on Fr\'echet manifolds such
that $A \in C^\oo(\C)$~\cite{km,fr,bfr}. We will not enter into such
details, and simply declare functions like $A$ to be in $C^\oo(\C)$. The
class of functions on $\C$ defined by an equation like~\eqref{eq:locdef}
will be referred to as \emph{local functionals}.

On the other hand, given an element $A\in C^\oo(\C)$, we can define a
notion of spacetime support that can be attributed directly to $A$. If
$A$ is local and comes from a horizontal density $\alpha$, there will of
course be a relation between these two notions of support. More
precisely, we define~\cite[Eq.5.22]{bf-notes}
\begin{multline}\label{eq:supp}
	\supp A = \{
		x \in M \mid \forall \text{ open } U \ni x ~
		\exists \Phi,\Psi \in \C \colon \\
			\supp \Psi \sse U \text{ and } A[\Phi + \Psi] \ne A[\Phi] \} ,
\end{multline}
which is always closed. In words, for any point $y\in M$ outside $\supp
A$, there is a sufficiently small neighborhood $V\ni y$ so that any
perturbation $\Psi$ of the argument of $A[\Phi]$ with $\supp \Psi \sse
V$ must leave the numerical value of $A$ unchanged, that is, $A[\Phi +
\Psi] = A[\Phi]$.  In other words, $A[\Phi]$ does not depend on the
value of $\Phi$ in some neighborhood of $y$.

As mentioned above, we can give a precise relation between the spacetime
support of a horizontal density and that of the corresponding local
functional. Recall the Euler-Lagrange derivative $\delta_{EL}[\alpha]$
of a horizontal density $\alpha$ defined by Equation~\eqref{eq:EL-def}.
Since $\delta_{EL}[\alpha]$ is not strictly speaking a horizontal
density, we extend to it the notion of spacetime support so that
$\supp_M \delta_{EL}[\alpha]$ is the union of the spacetime supports
$\supp_M \delta_{EL}[\alpha]_a$ of its components.
\begin{lem}\label{lem:supp}
Let $\alpha \in \Omega^{n,0}(F)$ be a horizontal density with compact
spacetime support and $A[\Phi] = \int_M (j^k\Phi)^* \alpha$. Then
\begin{equation}
	\supp_M \delta_{EL}[\alpha] \sse \supp A \sse \supp_M \alpha .
\end{equation}
\end{lem}
\begin{proof}
The second inclusion is trivial, because $(j^k(\Phi+\Psi))^*\alpha =
(j^k\Phi)^*\alpha$ whenever $\supp \Psi$ is outside of $\supp_M\alpha$,
since the restriction of both sides of the equality to $\supp\Psi$ is
simply zero. The rest, namely $\supp A \sse \supp_M \alpha$, follows
from the defining Equation~\eqref{eq:supp}.

On the other hand, suppose that $p \in \supp \delta_{EL}[\alpha] \sse
J^kF$. Then, we can always find a section $\Phi\in \C$ such that
$j^k\Phi(x) = p$, where $x = \pi^k(p)\in \supp_M \delta_{EL}[\alpha]$.
Since by construction $(j^k\Phi)^*\delta_{EL}[\alpha](x) \ne 0$, for
each open $U\ni x$ there must exist a (without loss of generality
compactly supported) $\Psi\in \C$ with $\supp\Psi \sse U$ such that
\begin{equation}
	\int_M (j^k\Phi)^*\delta_{EL}[\alpha]_a(x) \Psi^a(x) \ne 0 .
\end{equation}
Therefore, by continuity in $t$, the formula in
Equation~\eqref{eq:EL-def} tells us that there must exist a $t\ne 0$,
however small, such that $A[\Phi + t\Psi] \ne A[\Phi]$. That concludes
the proof that $\supp_M \delta_{EL}[\alpha] \sse \supp A$.
\end{proof}

\section{Generalized local observables}\label{sec:glocobsv}
A precise notion of a local functional on the space $\C$ of field
configurations on a field bundle $F\to M$ was given in
Section~\ref{sec:locobsv}. This notion is plenty sufficient to identify
a rich set of observables in the usual relativistic field theories,
including gauge theories like Maxwell electrodynamics and Yang-Mills
theory, but notably excluding gravitational theories like GR or GR with
matter fields. The reason gravitational theories are different is
because, as will be discussed in Section~\ref{sec:diffinv}, the
intersection between the space of local functionals and gauge invariant
functionals on $\C$ is trivial (it consists only of constant functions).
On the other hand, we can relax the above notion of locality in a
precise way, without sacrificing much in the way of the physical
motivation that lead to it, such that the new class of generalized local
functionals does admit a rich set of gauge invariant observables even in
gravitational theories. We discuss this precise generalized notion of
locality below and leave the applications to gravitational theories to
Section~\ref{sec:diffinv}.

The two main properties of local functionals that we would like to relax
are the (a) global domain of definition and (b) field independent
compactness of support. We explain both of these properties and how they
could be relaxed below.

Any element $A\in C^\oo(\C)$, by definition, gives a well-defined value
$A[\Phi]$ for any $\Phi\in \C$. That is, the domain of definition of $A$
is all of $\C$ (it is \emph{global}). Imagine, on the other hand, that
$A$ is defined only on a subset $\U\sse \C$. Could then $A$ still play
the role of a physically meaningful observable? The answer is a
qualified \emph{yes}, provided $\U$ is sufficiently large, for example
an open set. Such a restriction may be necessary if, for instance, we
have precise control only over solutions that are not too distant from a
reference solution,%
	\footnote{An example of this kind is the celebrated result of
	Christodoulou and Klainerman~\cite{ck} on the stability of Minkowski
	space in GR. Their result essentially constructs an open neighborhood
	$\U$ of the Minkowski metric on the phase space $\cP$ of GR on
	$\mathbb{R}^4$ with asymptotically flat boundary conditions. On the
	other hand, we still have very little information about $\cP$ outside
	that neighborhood.} %
some $\Phi\in \U$. At the classical level, having $A$ and $B$ defined on an
open neighborhood $\U \ni \Phi$ is sufficient to compute their Poisson
brackets%
	\footnote{Strictly speaking, Poisson brackets are expected to be
	defined only upon restriction to the phase space $\cP\sso \C$. However,
	it is sometimes possible to lift Poisson brackets even to $\C$. This
	will be discussed in more detail in Section~\ref{sec:pois}} %
at $\Phi$ because that involves only local, differential operations.
Perturbative QFT about $\Phi$ will also not be sensitive to anything
outside an arbitrary neighborhood. Eventually, a non-perturbative
formulation of a QFT would likely require observables to be globally
defined. However, even then, we are likely to be interested in quantum
states that (e.g.,\ in a phase space formulation of quantum theory)
would assign negligible weight to solutions outside a neighborhood $\U$
of some reference solution $\Phi$. To accommodate such an eventual
situation, we could globalise the domain of definition of $A\in
C^\oo(\C)$ by extending it in an arbitrary, though controlled way, to
all of $\C$ using standard geometric tools, like the Tietze extension
and Steenrod-Wockel approximation theorems~\cite{wockel}.

Given that we would like the domain $\U\sso \C$ of a generalized local
functional to be open, it is important to reflect on the topology that
we use on $\C$. Technical details on various topologies on function
spaces can be found in the references~\cite{hirsch,km}. It was stated in
the Introduction that it is conventional to endow $\C = \Secs(F)$ with the
\emph{Whitney weak topology}, whose open sets are generated by those of
the form
\begin{equation}
	\U^k_{K,U} = \{ \Phi \in \Secs(F) \mid j^k\Phi(K) \sse U \} ,
\end{equation}
where $k\ge 0$, $K\sse M$ is compact and $U\sse J^kF$ is open. The big
disadvantage of the weak topology is that its neighborhoods cannot control
the behavior of a section outside of a compact subset of the spacetime
$M$, as we will need to do in the sequel. However, except in some cases
when boundaries are present, the spacetimes that are of physical
interest are non-compact. For example, any globally hyperbolic spacetime
must be of the form $M \cong \mathbb{R}\times \Sigma$. An alternative
topology is the \emph{Whitney strong topology}, whose open sets are
generated by those of the form
\begin{equation}
	\U^k_U = \{ \Phi \in \Secs(F) \mid j^k\Phi(M) \sse U \} ,
\end{equation}
where $k\ge 0$ and $U\sse J^kF$ is open. The big disadvantage of the
strong topology is that it is incompatible with the structure of a
topological vector space on $\C$ (multiplication by scalars fails to be
continuous), let alone a Fr\'echet or any other kind of manifold
structure. Note, though, that since our manifolds can be exhausted by
compact sets, any open set in the strong topology is at worst a
$G_\delta$ set in the weak topology (a countable intersection of open
sets). Fortunately, there are many intermediate topologies between the
weak and the strong that both allow a Fr\'echet structure and control
the behavior of sections on all of $M$. One example is a variation on
the strong topology that allows only those open $U\sse J^kF$ that have
``uniform'' vertical size over $M$ with respect to some connection, such
as one induced by an auxiliary Riemannian metric.  Another possibility
is to add a compactifying boundary to $M$ and restrict our attention only
those sections that extend in some nice way to the boundary,%
	\footnote{Perhaps the simplest implementation of this idea is to
	consider a piece of a globally hyperbolic spacetime that is bounded by
	two compact Cauchy surfaces as a compact spacetime in its own right
	with the future and past Cauchy surfaces as its boundaries.} %
then using the weak topology on that subspace with respect to
the compactified spacetime $M$. However, it does not seem that there is
an a priori canonical choice of such an intermediate topology and that
the choice must be made in a way that is compatible with the behavior of
solutions of the equations of motion of the theory. Note that a similar
discussion, and in a related context, can be found in Section~5.2.1
of~\cite{kh-caus}.

Being pragmatic, we stick to the Whitney strong topology for the
remainder of this note, despite its drawbacks. The working hypothesis is
that the results that will be found in the sequel, and the methods used
to obtain them, will naturally generalize to the appropriate choice of
intermediate topology.

Next, having taken the liberty of considering functionals that are
defined only on open subsets $\U \sse \C$, let us consider the
difference between the spacetime supports of a functional $A\in
C^\oo(\C)$ and its restriction $A|_\U \in C^\oo(\U)$. We can reasonably
define $\supp A|_\U$ by replacing $\C$ with $\U$ in the
definition~\eqref{eq:supp}. The logical quantifiers are arranged such
that $\supp A|_\U \sse \supp A$. In fact, we can define the even finer
notion of \emph{spacetime support at $\Phi$ with respect to $\U$} given by
\begin{multline}\label{eq:lsupp-u}
	\supp_\Phi A|_\U = \{ x \in M \mid \forall \text{ open } U \ni x ~
		\exists (\Phi+\Psi) \in \U \colon \\
		\supp \Psi \sse U \text{ and } A[\Phi+\Psi] \ne A[\Phi] \} ,
\end{multline}
which is also always closed. A further refinement is the notion of
\emph{spacetime support at $\Phi$} given by
\begin{equation}\label{eq:lsupp}
	\supp_\Phi A = \bigcap_{\U} \supp_\Phi A|_\U ,
\end{equation}
with the intersection taken over all open neighborhoods $\U \ni \Phi$
such that $A$ is defined on $\U$. The distinction is that while $\supp_\Phi
A|_\U$ depends on the domain $\U$, $\supp_\Phi A$ only depends on
the germ of $A$ at $\Phi$. Then $\bigcup_{\Phi\in \U} \supp_\Phi A \sse
\supp A|_\U$ and $\supp A|_\U = \bigcup_{\Phi\in \U} \supp_\Phi A|_\U$.
So, clearly, $\supp A|_\U$ may fail to be compact, even if each
individual $\supp_\Phi A|_\U$ or $\supp_\Phi A$ is.

For $A|_\U$ to be IR regularizing, as discussed in the Introduction, it
suffices that the spacetime supports $\supp_\Phi A$ be compact for
each $\Phi \in \U$. Thus, the much stronger condition of compact $\supp
A|_\U$ for an observable $A|_\U$, while obviously sufficient for IR
regularity, is not necessary. Such a relaxation of the requirements on
the field-dependent spacetime support of observables was previously
considered in~\cite[Sec.5.3.5]{kh-caus} (see also~\cite{sharapov}).

At the level of local functionals, we can relax the notion of locality
given in Section~\ref{sec:locobsv} in the following way. Let $\Phi\in
\C$ be a field configuration and $\alpha \in \Omega^{n,0}(F)$ be a
horizontal density such that the intersection $j^k\Phi(M) \cap \supp
\alpha \sse J^kF$ is compact. Then we call the functional
\begin{equation}\label{eq:gloc-int}
	A[\Psi] = \int_M (j^k\Psi)^* \alpha
\end{equation}
a \emph{generalized local functional (or observable) at $\Phi$}. The
following result makes precise the way in which the properties of the
functional $A$ fit with the preceding discussion.
\begin{thm}\label{thm:loc}
With $\Phi$ and $\alpha$ as above, there exists an open $\U\sse \C$ (in
the strong topology) with $\Phi\in \U$ such that, for all $\Psi \in \U$,
the integral in~\eqref{eq:gloc-int} is convergent and both $\supp_\Psi
A|_\U$ and $\supp_\Psi A$ are compact.
\end{thm}
\begin{proof}
Pick a compact neighborhood $Q$ of $K = \pi^k (j^k\Phi(M) \cap \supp
\alpha)$ and an open neighborhood $U\sso J^kF$ of $j^k\Phi(M\setminus
Q)$ that does not intersect $\supp \alpha$. Let $\U\sse \C$ be the set
of all sections $\Psi\colon M\to F$ such that $j^k\Psi(M\sm Q) \sso U$.
Clearly, $\Phi\in \U$ and, by the definition of the Whitney strong
topology, $\U$ is open. By construction, for any $\Psi \in \U$, we have
$(j^k\Psi)^* \alpha = 0$ on $M\sm Q$. This means that $\supp [
(j^k\Psi)^*\alpha] \sse Q$ and is itself compact (by virtue of being a
closed subset of a compact set) and hence the integral defining
$A[\Psi]$ is convergent. Finally, from the definition of $\U$, any point
$x\in M\sm Q$ has a neighborhood $V\sse M\sm Q$ such that any $\Delta$
that has $\supp \Delta \sse V$ and with $\Psi + \Delta \in \U$ and must
satisfy $j^k(\Psi+\Delta)^*\alpha = 0$ on $V$ and hence $A[\Psi +
\Delta] = A[\Psi]$. Therefore $\supp_\Psi A|_\U \sse Q$ and hence is
itself compact. Its subset $\supp_\Psi A \sse \supp_\Psi A|_\U$ is
closed and hence also compact.
\end{proof}

\section{Gauge invariance and local observables in gravitational theories}\label{sec:diffinv}

GR is the theory of a Lorentzian metric field $G$,
so that $F = S^2T^*M$, with the equation of motion (Einstein equation)
specified by the Einstein-Hilbert Lagrangian, $\mathcal{L}[G] = R[G]\,
\mathrm{vol}_G$, where $R[G]$ is the Ricci scalar and $\mathrm{vol}_G$
is the metric volume form. This Lagrangian also determines the gauge
symmetries of the theory, which consist of diffeomorphisms of $M$ acting
by pullback on metrics, $G \mapsto \chi^* G$ for a diffeomorphism
$\chi\colon M\to M$. Thus, the physical (or reduced) phase space of GR
is the quotient $\bP = \cP / \G$, where $\cP\sso \C$ is the set of
solutions of Einstein equations (usually also taken to be globally
hyperbolic) and $\G$ is the group of gauge transformations
(diffeomorphisms of $M$). The observables that we are really interested
in are those that constitute the algebra $C^\oo(\bP)$.  As before, it is
convenient to use the quotient map $\cP \to \bP$ to identify $C^\oo(\bP)
\sso C^\oo(\cP)$ with those observables that are invariant under the
action of the group $\G$ of gauge transformations.  We refer to any
element $A \in C^\oo(\cP)$, or $C^\oo(\C)$, as a \emph{gauge invariant
observable (or functional)} if it is left invariant by the action of
$\G$, that is, $A[\chi^* G] = A[G]$ for any diffeomorphism $\chi\colon M
\to M$. For our purposes, a \emph{gravitational theory} is a field
theory that involves a metric tensor $G$ (though possibly other fields
as well) and has the diffeomorphism group as the group $\G$ of gauge
transformations.  Clearly, GR is the representative example of a
gravitational theory, but GR coupled to matter fields also falls into
the same category. We will only consider pure GR below, but the
discussion will also apply to more general gravitational theories.

It is a well-known folk result that GR does not have any \emph{local}
and \emph{gauge invariant observables} in the standard sense of locality
discussed in Section~\ref{sec:locobsv}. However, the main observation of
this note is that there in fact do exist local and gauge invariant
observables in the generalized sense discussed in
Section~\ref{sec:glocobsv}. The non-existence argument is pretty
straight forward. Let $\alpha$ be a horizontal density on $k$-jets with
$\supp_M \alpha$ compact and hence $A[G] = \int_M (j^k G)^* \alpha$ a local
observable. A diffeomorphism $\chi\colon M\to M$ acts on it as
\begin{equation}
	(\chi A)[G] = A[\chi^* G] = \int_M (j^k (\chi^* G))^* \alpha
%		= \int_M \chi^* (j^k G)^* (p^k \chi^*)^* \alpha
		= \int_M (j^k G)^* \left[ (p^k \chi^*)^* \alpha \right],
\end{equation}
where $p^k \chi^* \colon J^kF \to J^kF$ is the natural $k$-jet
prolongation of the pullback action of a diffeomorphism on metrics
$\chi^*\colon F\to F$. Clearly, the spacetime support of $\alpha$
transforms as $\supp_M \left[ (p^k\chi^*)^*\alpha \right] = \chi(\supp_M
\alpha)$. Thus, by Lemma~\ref{lem:supp}, the $\supp \chi A$ moves around
on $M$ under the action of diffeomorphisms. So, since we can choose
$\chi$ such that $\supp A$ and $\supp \chi A$ do not coincide, the
functionals $A$ and $\chi A$ themselves cannot coincide. In particular,
no observable $A$ can be gauge invariant if its spacetime support is
different from $M$ (diffeomorphisms act on $M$ transitively). Spacetime
manifolds of physical interest are never compact, hence no local
observable (with, by definition from Section~\ref{sec:locobsv}, compact
spacetime support) can be gauge invariant. Colloquially, this is phrased
as follows: \emph{gauge transformations of gravitational theories move
spacetime points}. This property is in contrast with gauge theories of
Maxwell or Yang-Mills type, where gauge transformations leave intact the
spacetime support of observables, thus allowing local observables to be
gauge invariant.

We now give an explicit example of a functional that is both gauge
invariant and local in the generalized sense. Subsequently, we will outline
a general method for constructing more examples of a similar kind. Let
us restrict for the moment the dimension $\dim M = 4$. We will construct
a horizontal density $\alpha\in \Omega^{4,0}(F)$ on $J^3F$. Let
$W_{abcd} = W_{abcd}[G]$ and $\eps_{abcd} = \eps_{abcd}[G]$ denote
respectively the Weyl and Levi-Civita tensors of the metric $G$. Then,
define the dual Weyl tensor $W^*_{ab}{}^{cd} = W_{abc'd'}\eps^{c'd'cd}$
and also the following curvature scalars
\begin{equation}\label{eq:b-def}
\begin{aligned}
	b^1 &= W_{ab}{}^{cd} W_{cd}{}^{ab} , & \quad
	b^3 &= W_{ab}{}^{cd} W_{cd}{}^{ef} W_{ef}{}^{ab} , \\
	b^2 &= W_{ab}{}^{cd} W^*_{cd}{}^{ab} , & \quad
	b^4 &= W_{ab}{}^{cd} W_{cd}{}^{ef} W^*_{ef}{}^{ab} .
\end{aligned}
\end{equation}
We have essentially defined maps $b = (b^1,b^2,b^3,b^4) \colon J^kF \to
\mathbb{R}^4$, for any $k\ge 2$. We will also use the notation $(b^i)$
for the standard global coordinates on this target $\mathbb{R}^4$.
It is sufficient for us to take $k=3$ because we then want to define the
horizontal density $\beta = \h[\d b^1 \wedge \d b^2 \wedge \d b^3 \wedge
\d b^4] \in \Omega^{4,0}(F,3) \sso \Omega^{4,0}(F)$. Choose a point
$r\in \mathbb{R}^4$, and a function $f\in C^\oo(\mathbb{R}^4)$ with
compact support, such that $r\in \supp f$ but $\supp f$ does not
intersect any of the planes $b^i = 0$. Finally, we define the desired
horizontal density $\alpha = f(b) \beta \in \Omega^{4,0}(F)$, which
gives rise to the functional
\begin{equation}\label{eq:lgi-def}
	A[G]
	= \int_M (j^kG)^* \alpha \\
	= \int_M (j^kG)^* \left( f(b^1,b^2,b^3,b^4)\,
			\h[\d{b^1} \wedge \d{b^2} \wedge \d{b^3} \wedge \d{b^4}] \right) .
\end{equation}
By construction, $\alpha$ satisfies two important properties. First,
there is a non-empty open set $\U \sse \C$ (in the strong topology) such
that the form $(j^k G)^* \alpha$ is smooth and has compact support on
$M$ for any $G\in \U$. Thus, $A[G]$ is well-defined on $\U$ and hence
constitutes a generalized local observable in the sense of
Section~\ref{sec:glocobsv}.  The existence of such a domain $\U$ follows
from a general result that will be discussed in
Theorem~\ref{thm:generic} (see also the comments thereafter). Second,
$A[G]$ is invariant under the action of diffeomorphisms. That is,
$(p^k\chi^*)^* \alpha = \alpha$ for any diffeomorphism $\chi\colon M\to
M$, which implies $A[\chi^*G] = A[G]$ for any $G$ on which the defining
integral converges. The last invariance identity has to be used with a
little bit of care, in that it only makes sense when both $G$ and
$\chi^*G$ belong to $\U$, the domain of definition of $A$. Since, a
priori $\U$ is not guaranteed to be itself diffeomorphism invariant,
that condition may not be satisfied for an arbitrary $G\in \U$. One way
to get around this issue is to, very reasonably, declare $A$ to be
invariant under diffeomorphisms if $A[\chi^* G] = A[G]$ whenever both
$G, \chi^*G \in \U$. Another way is to simply enlarge the domain to $\U'
\supseteq \U$ to the smallest diffeomorphism invariant domain that
contains $\U$. Clearly, if $A$ is well defined on $\U$ it is also well
defined on $\U'$. A note of caution for second approach: while
$\supp_\Phi A$, for any $\Phi\in \U$, is not altered by extending $A$
from $\U$ to $\U'$, the inclusion $\supp_\Phi A|_\U \sse \supp_\Phi
A|_{\U'}$ may be strict, with $\supp_\Phi A|_{\U'}$ possibly failing to
be compact even if $\supp_\Phi A|_\U$ is.

In other words $A|_\U \in C^\oo(\U)$ is a local and gauge invariant
observable in the generalized sense of Section~\ref{sec:glocobsv}.

The idea of using the curvature scalars $b^i$ to define observables in
pure gravity goes back to the proposal of Komar and
Bergmann~\cite{bg0,bg1}.  However, these authors, as well as many
subsequent ones who came back to this idea (see~\cite{tamb} and
references therein), intended to use $b^i$ as independent coordinates
and simply express all other fields in terms of them. However, the
resulting observables were often too singular in the sense discussed in
the Introduction, since they would correspond to something like
replacing our test function $f$ with a $\delta$-distribution. On the
other hand, our addition of the integral and the smooth compactly
supported function $f$ and the definition~\eqref{eq:lgi-def} provides
the diffusion of UV singularities and the IR regularization, again
discussed in the Introduction, that are needed in the contexts of QFT
and classical Poisson structure.

The key ingredients in the above construction were the facts that we
could choose the horizontal density $\alpha\in \Omega^{n,0}(F)$ to be
invariant under the prolonged action of diffeomorphisms on $J^kF$ and
the fact that we could choose such an $\alpha$ to have support on $J^kF$
that intersects compactly the image of the prolongation $j^kG(M) \sso
J^kF$ of a certain metric $G$. Natural questions arise. Are there
more local and gauge invariant observables that could be defined in the
same way? Are there sufficiently many such observables to separate
points%
	\footnote{A set of functions \emph{separates} the points of a space
	if, for each pair of points, there exist at least one function that
	takes on different values at these points.} %
on the physical phase space $\bP$ of GR?

The general mathematical context in which the answers must be sought is
known as \emph{differential invariant theory}~\cite{kj,op,kl}. Classical
invariant theory is concerned with identifying functions on a $\G$-space
(a space with an action of a group $\G$) that are invariant under the
$\G$-action, these are the usual \emph{invariants}. On the other hand,
differential invariant theory, is concerned with fiber preserving group
actions (more generally \emph{pseudogroup} or \emph{groupoid} actions)
on the total space of a bundle, like our field vector bundle $F\to M$,
and the actions induced on $J^kF \to M$ by prolongation. Then,
\emph{differential invariants (of order $k$)} are functions on $J^kF\to
M$ that are invariant under the group action. For our purposes, the
field bundle of metrics is $F = S^2T^*M$ and the group is $\G=\Diff(M)$,
consisting of diffeomorphisms $\chi \colon M\to M$, and acting by
pullback $\chi^*\colon F\to F$. Differential invariants are then
precisely the so-called \emph{curvature scalars}, that is, scalar
functions tensorially constructed out of the metric, the Riemann
curvature tensors and its covariant derivatives. For example, the $b^i$
defined in Equation~\eqref{eq:b-def} are differential invariants of
order $k=2$. There are two ways of looking at differential invariants:
\emph{algebraically} and \emph{geometrically}. Most structural results
are proven from the algebraic perspective. On the other hand, it is
easier to see from the geometric perspective how to construct local
gauge invariant observables similar to the example of
Equation~\eqref{eq:lgi-def}.

The main structural algebraic result that we would like to mention is
the so-called \emph{Lie-Tresse} theorem, which dates back to the end of
the 19th century, but was established in its global form only rather
recently (see~\cite{kl} and the references therein). This theorem is an
analog of the \emph{finite generation} results, originally due to
Hilbert, in classical invariant theory~\cite{olv}. For differential
invariants, in addition to algebraic operations, we also need to allow
differentiation to generate differential invariants of arbitrary orders
from finite data. Before stating the result, let us recall the geometric
formulation of differential equations in terms of jets,
cf.~\cite[Apx.B]{kh-peierls} for more details and references. A
differential equation of order $m\ge 0$, is usually specified
in \emph{equational form}, $P[\psi] = 0$, where $P\colon \Secs(F) \to
\Secs(E)$ a possibly non-linear differential operator of order $m$ that
takes sections of a bundle $F\to M$ as arguments and output sections of
some other vector bundle $E\to M$. Essentially equivalently, we can
specify a differential equation of order $m$ as submanifold
$\mathcal{E}\sse J^mF$. Roughly, the set of all $m$-jets that satisfy
$P=0$ constitutes the subset $\mathcal{E}$ and inversely, any bundle map
$P\colon J^mF \to E$ that is zero only on $\mathcal{E}\sse J^mE$ defines
the corresponding differential operator. A particular example could be
$\mathcal{E} = J^mF$, which corresponds to the trivial equation $0=0$. A
differential equation $\mathcal{E}\sse J^mF$ has natural prolongations
$\mathcal{E}^{(k)} \sse J^kF$ for all $k\ge m$, which corresponds to
taking into account all equations of the form $\del_{i_1} \cdots
\del_{i_{k-m}} P[\psi] = 0$ implied by $P[\psi] = 0$. The following
result is a rough restatement (sufficient for the purposes of this note)
of the precise results of Theorems~1 and 2 of~\cite{kl}.

\begin{prop}[Lie-Tresse]\label{prp:lt}
Consider a differential equation $\mathcal{E}\sse J^kF$ with gauge
symmetry,%
	\footnote{In the language of~\cite{kl}, this means that the equation
	is invariant under a pseudogroup action.} %
defined on a field bundle $F\to M$, with the action of gauge symmetries
naturally prolongued to $J^kF \to M$. Assume that the equation and the
gauge symmetry action satisfies a specific global algebro-geometric
regularity condition (which is in fact satisfied by GR with
diffeomorphisms as gauge symmetries). Then, there exists a finite order
$l\ge 0$, a finite number of differential invariants (those left
invariant by gauge transformations) $I_j$ on $J^lF$, and a finite number
of invariant differential operators $D_i$ (such an operator acting on an
invariant yields another invariant) such that any polynomial differential
invariant of an arbitrary order $k\ge 0$ can be expressed as a
polynomial in the generators $I_j$, possibly repeatedly differentiated
by the $D_i$. Finally, for arbitrary order $k\ge 0$, the differential
invariants separate the orbits of the gauge symmetry on a dense open
subset $\mathring{\mathcal{E}}^{(k)} \sse \mathcal{E}^{(k)}$ consisting
of generic orbits.
\end{prop}

% NB: The basic idea underlying this article dawned on my after reading
% paper {gns} around the weekend of 25 Jan 2014.
More geometrically, we can look at differential invariants as follows.
Consider the quotient spaces%
	\footnote{At the moment, we are not making a notational distinction,
	but we are really only interested in the subset of $J^kF$
	corresponding to the jets of Lorentzian metrics, thus excluding
	degenerate metrics and metrics of other signatures.} %
$\M^k = J^kF / \Diff(M)$, known as the \emph{moduli spaces of $k$-jets
of metrics on $M$}~\cite{gns}, with the projections denoted by
$\mu_k\colon J^kF \to \M^k$. Clearly, differential invariants are
precisely the smooth functions on $J^kF$ that come from the pullback of
continuous functions on $\M^k$, those that belong to $C^\oo(J^kF) \cap
\mu_k^*[C(\M^k)]$. If $\M^k$ were a manifold, it would be sufficient to
consider $C^\oo(\M^k)$ instead of $C(\M^k)$. However, while $\M^k$ is
well-defined as a topological space, it is only a manifold on a dense
open subset~\cite{gns}, say $\rM^k \sso \M^k$. Outside $\rM^k$, $\M^k$
contains orbifold-type singularities, which correspond to jets of
metrics admitting non-trivial isometries. A further complication is that
$\M^k$ is in general not Hausdorff. This means that there exist jets of
metrics that cannot be distinguished by continuous scalar curvature
invariants alone. This phenomenon is particular to Lorentzian (and other
pseudo-Riemannian) metrics and is absent when consideration is
restricted to only Riemannian metrics. The failure of the Hausdorff
property can be traced back to the non-compactness of the orthogonal
group in Lorentzian signature~\cite{gns}.

To connect the algebraic and geometric points of view, consider
Einstein's equations prolonged to an arbitrary order $k\ge 2$ and
represented as a submanifold $\mathcal{E}^{(k)} \sse J^k F$. Clearly,
$\mathcal{E}^{(k)}$ is invariant under diffeomorphisms and so projects
to $\mu_k \colon \mathcal{E}^{(k)} \to \mathcal{R}^k\sse \M^k$, with
$\mathring{\mathcal{R}}^k = \mathcal{R}^k \cap \rM^k$ a submanifold of
$\rM^k$. The polynomial differential invariants mentioned in
Proposition~\ref{prp:lt} are then functions on $\mathcal{R}^k$ and in
fact separate the points of $\mathring{\mathcal{R}}^k$ and, by the
Stone-Weierstrass theorem, generate $C^\oo(\mathring{\mathcal{R}}^k)$ by
limits uniformly converging on compact sets. 

Now we come to the main observation that prompted this note. The
connection between differential invariants and local observables in the
generalized sense of Section~\ref{sec:glocobsv} is most clearly seen with
the help of the manifold $\rM^k$. Namely, consider an $n$-form $\beta
\in \Omega^n(\rM^k)$ with compact support and the horizontal density
$\alpha\in \Omega^{n,0}(F)$ obtained by the horizontal projection of the
pullback of $\beta$, $\alpha = \h[\mu_k^*\beta]$. Letting $\U \sso \C$
be the subset of all metrics $G\colon M\to F$ such that $j^kG(M) \cap
\supp \alpha$ is compact, we can define a local and gauge invariant
observable with domain of definition $\U$ by the usual formula
\begin{equation}\label{eq:obsv-def}
	A[G] = \int_M (j^k G)^* \alpha .
\end{equation}
It is clearly gauge invariant, since by construction $(p^k\chi^*)^*
\alpha = \alpha$. Further, it is clearly local in the generalized sense
of Section~\ref{sec:glocobsv}, provided that $\U$ is open and non-empty.
These properties do hold because of the following
\begin{thm}\label{thm:generic}
Given a non-empty compact set $K \sso \rM^k$, there exists a metric
$G\in \C$ such that $(\mu_k\circ j^kG)^{-1}(K)\sse M$ is non-empty and
compact.  Further, such a metric $G$ has an open neighborhood $\U\sso
\C$ (in the strong topology) such that $\mu_k\circ j^kH(M) \cap K$ is
compact for each $H\in \U$.
\end{thm}
\begin{proof}
First, we deal with the statement about existence. Let us ignore for the
moment issues that might arise from non-trivial topology of $M$ and
assume that $M\cong \mathbb{R}^n$, with some fixed global coordinate
system. Let $\eta\colon M\to F$ be the standard Minkowski metric in
those global coordinates. Take a point $r\in K \sso \rM^k$, a point $x\in M$
and an open neighborhood $U\sso M$ of $x$ with compact closure. By
construction, there is a jet $p\in J^k_xF$ such that $\mu_k(p) = r$.
Consider the closed set $Q = (M\sm U) \cup \{x\}$.  Define $G^k_Q \colon
Q\to J^kF$ so that $G^k_Q(x) = p \in J^k_xF$ and $G^k_Q(y) = j^k_y\eta
\in J^k_yF$ for any $y\ne x$. By the Whitney extension
theorem~\cite[\textsection 22]{km}, there exists a metric $G\in \C$ such
that $j^kG(x)|_Q = G^k_Q$, which we can choose to be everywhere
Lorentzian (non-degenerate). Thus, $\mu_k\circ j^kG$ and $K$ have at
least the point $r$ in common. On the other hand, by construction, the
pre-images $(\mu_k\circ j^kG)^{-1}(K) \sso (\mu_k\circ j^kG)^{-1}(\rM^k) \sso M$ must be contained in
$\bar{U}$, which is compact. Hence, the pre-image of $K$ must be
compact, since it is closed and contained in $\bar{U}$. The same
argument can be adapted without much difficulty to the case when $M$ has
more complicated topology.

Second, we deal with the statement about an open neighborhood of $G\in
\C$, which was constructed above. The following argument echos the proof
of Theorem~\ref{thm:loc}. We will define $\U = \{ H\in \C \mid j^kH(M)
\sso U \}$, for some to be determined open neighborhood $U \sso F$ of
$j^kG(M)$. Obviously $G\in \U$ and $\U$ would be open in the strong
topology. We build $U$ as the pre-image of an open set $V\sse M \times
\M^k$ with respect to the map $(\pi^k,\mu_k)\colon J^kF \to M\times
\M^k$. If $V$ is an open neighborhood of the graph of $\mu_k \circ
j^kG\colon M \to \M^k$, then $U$ is an open neighborhood of $j^kG(M)$.
The way we constructed $G$ above, the intersection $I$ of the set
$M\times K$ with the graph of $\mu_k\circ j^kG$ is compact. Take an open
neighborhood $V'$ with compact closure of $I$ and let $V = M\times
(\M^k\sm K) \cup V'$. Thus, if $H\in \U$, the intersection of the graph
of $\mu_k\circ j^kH$ with $M\times K$ must be confined to $V'$, which
has compact closure, and hence be compact. The last statement is
equivalent to the pre-image $(\mu_k\circ j^kH)^{-1}(K)\sso M$ being
compact, which concludes the proof.
\end{proof}
Note that a direct application of the above theorem to the compact
$\supp f$ appearing in the definition of the functional $A[G]$ given by
Equation~\eqref{eq:lgi-def}, interpreted as a subset of $\rM^2$,
establishes the claimed existence of a non-empty open domain $\U\sse
\C$, making $A|_\U$ a generalized local observable.

While we have concentrated on the case of gravitational theories, whose
group of gauge transformations consists of diffeomorphisms, this method
of defining gauge invariant local observables happens to reproduce the
set of local observables for theories without gauge symmetries (the
group of gauge symmetries is trivial) and those with gauge theories with
gauge transformations that do not \emph{move points}. Examples of the
latter include the Maxwell and Yang-Mills theories. In the Maxwell
theory, the basic differential invariant is the field strength. In the
Yang-Mills case, the basic differential invariants are the compositions
of the Lie algbra valued curvature forms composed with invariant
polynomials on the Lie algebra. Smearing these basic invariants (or
derivatives thereof) with compactly supported test functions reproduces
the well-known standard local and gauge invariant observables in these
theories~\cite{bbh}.

We conclude this section by coming back to this natural question: are
there enough local and gauge invariant observables in GR to separate the
points of $\C$? In a sense, the answer is No, because we have already
discussed above the fact that certain metrics cannot be distinguished by
local curvature scalars. Further, some metrics may be resistant to
belonging to the domain of definition $\U$ of any generalized local
observable $A|_\U$. This may happen when $M$ is non-compact and a metric
$G$ possesses a region $U\sse M$ such that nearly isometric copies of
$G|_U$ repeat infinitely often throughout $M$ (a kind of almost periodic
property). There is essentially no obstacle to engineering a gauge
invariant local density $\alpha$ on $J^kF$ such that $(j^kG)^*\alpha$
has compact support in $U$, but it will likely also have support within
any region nearly isometric to $G|_U$, thus making the integral over $M$
ill defined. However, these are the only obstacles. We need to introduce
a natural but somewhat technical condition on metrics that avoid these
difficulties.

First, we say that a map $\nu\colon M \to N$ is \emph{image proper}%
	\footnote{Cf.~\cite[Exr.2.4.13]{hirsch}, where this concept is used
	but not named.} %
if there exists an open set $N_0 \sse N$ such that $\nu(M) \sse N_0$ and
$\nu\colon M \to N_0$ is \emph{proper} (the pre-image of any compact set
is compact). Any proper map is image proper, since we can just choose
$N_0 = N$. On the other hand, any embedding is image proper, even if it
is not proper, with any tubular neighborhood fulfilling the role of
$N_0$. Let us say that two metrics $G_1,G_2\in \C$ can be
\emph{distinguished by curvature scalars} if there exists a $k\ge 0$
such that $\gamma_i = \mu_k \circ j^kG_i \colon M \to \M^k$ are image
proper and the images $\gamma_1(M) \cap \rM^k$ and $\gamma_2(M) \cap
\rM^k$ do not coincide as subsets of $\rM^k$.
\begin{thm}\label{thm:sep}
For any two metrics $G_1,G_2\in \C$ that can be distinguished by
curvature scalars, there exists a local functional $A[G]$ defined on a
domain $\U\sse \C$ (open in the strong topology) such that both
$G_1,G_2\in \U$ and $A[G_1] \ne A[G_2]$.
\end{thm}
\begin{proof}
By hypothesis, there is a $k\ge 0$ and a point $r\in \rM^k$ such that,
say, $r\in \gamma_1(M) = \mu_k\circ j^kG_1(M)$ but $r\not\in \gamma_2(M)
= \mu_k\circ j^kG_2(M)$. Take a $\beta\in \Omega^n(\rM^k)$ with compact
support such that $r\in \supp\beta$ but $\gamma_2(M) \cap \supp\beta =
\varnothing$.  Let $A[G] = \int_M (\mu_k\circ j^kG)^* \beta$. Since the
map $\gamma_1$ is image proper, we can always choose $\beta$ so that
$\supp \beta$ is small enough to have compact intersection with
$\gamma_1(M)$ and so that $A[G_1] \ne 0$.  On the other hand, by
construction, $A[G_2] = 0$. Finally, since both $\gamma_i(M)$ have
compact intersection with $\supp\beta$ (one of the intersections being
empty), by Theorem~\ref{thm:generic}, there exist (in the strong
topology) open neighborhoods $\U_1$ and $\U_2$ of $G_1$ and $G_2$,
respectively, such that $\mu_k\circ j^kG(M)\cap \supp\beta$ is also
compact for each $G\in \U = \U_1 \cup \U_2$. Clearly, $A[G]$ is well
defined on $\U$ and $G_1,G_2 \in \U$.
\end{proof}

\section{Linearization and Poisson brackets}\label{sec:pois}
Once a class of gauge invariant observables has been defined, as was
done in Section~\ref{sec:diffinv}, we would like to compute Poisson
brackets between them. In general, neither the product nor the Poisson
bracket of two local observables is a local observable (instead it
is \emph{bilocal}, with distributional smearing in case of the Poisson
bracket) and the same is true for local observables in the generalized
sense. It is an important and non-trivial question to decide on a
minimal physically reasonable class of observables that is closed both
under multiplication and Poisson brackets. The answer is essentially a
class of multilocal observables with distributional smearings, which
satisfy a certain \emph{microlocal spectral condition}, which is
discussed in more detail in~\cite{bf-notes,bfr}. Below, we shall not be
concerned with these details and instead content ourselves with a gauge
invariant formula for the Poisson bracket of two local and gauge
invariant observables.

As discussed extensively in~\cite{kh-caus,kh-peierls}, what is usually
known as the canonical Poisson bracket on the physical phase space $\bP$
can be equivalently expressed using the so-called \emph{Peierls formula}
(or \emph{Peierls bracket}). The Peierls formula actually defines a
Poisson bracket not only on $C^\oo(\bP)$, but also extends it to
$C^\oo(\cP)$ and even $C^\oo(\C)$. This extension is not unique and is
influenced, for instance, by the choice of gauge fixing. However, the
restriction of the formula to $C^\oo(\bar{\cP})$ is unique. The
computation of the value of the Poisson bracket $\{A,B\}[G]$ of
arbitrary observables $A$ and $B$ at a particular point (or gauge
equivalence class of field configurations) $G\in \bP$ of a non-linear
field theory reduces to the computation of the Poisson bracket of linear
observables $\dot{A}_G$ and $\dot{B}_G$ in the linear theory obtained by
linearization about $G$. Consider the linearized perturbation $H$ of the
metric $G$. The relation between non-linear observables and linearized
observables is
\begin{equation}
	A[G+\lambda H] = A[G] + \lambda \dot{A}_G[H] + O(\lambda^2).
\end{equation}
In the case of a local observable $A[G] = \int_M (j^kG)^* \alpha$, the
linearized observable is also local, $\dot{A}_G[H] = \int_M
\dot{\alpha}[H]$, where $\dot{\alpha}$ is a density-valued differential
operator defined by
\begin{equation}
	(j^k(G+\lambda H))^*\alpha
		= (j^kG)^*\alpha + \lambda \dot{\alpha}_G[H] + O(\lambda^2).
\end{equation}
We can define similarly $B[G] = \int_M (j^kG)^*\beta$ and $\dot{B}_G[H] =
\int_M \dot{\beta}_G[H]$.

It is also useful to consider the formal adjoint differential operators
$\dot{\alpha}^*_G$ and $\dot{\beta}^*_G$ defined by the existence of
form-valued bidifferential operators $W_\alpha$ and $W_\beta$ such that
\begin{equation}
	f \dot{\alpha}_G[H] - \dot{\alpha}^*_G[f]\cdot H = \d W_\alpha[f,H]
	\quad \text{and} \quad
	f \dot{\beta}_G[H] - \dot{\beta}^*_G[f]\cdot H = \d W_\beta[f,H]
\end{equation}
for arbitrary $f\in C^\oo(M)$ and $H\in \Secs(F)$, with the adjoint
operators valued in the densitized dual bundle $\tilde{F}^* = F^*\otimes
\Lambda^n M$. Let $\U\sse \bP$ be a common domain on which $A$ and $B$
are defined and let $G\in \U$. Then, by the generalized locality property,
$\dot{\alpha}_G[H]$ and $\dot{\beta}_G[H]$ have compact support for
arbitrary $H$. It is then not hard to see that all of
$\dot{\alpha}^*_G[1]$, $W_\alpha[1,H]$, $\dot{\beta}^*_G[1]$ and
$W_\beta[1,H]$ will also have compact support for arbitrary $H$.
Therefore, an application of Stokes' lemma gives us the identities
\begin{equation}\label{eq:linobsv-equiv}
	\dot{A}_G[H] = \int_M \dot{\alpha}^*_G[1]\cdot H
	\quad \text{and} \quad
	\dot{B}_G[H] = \int_M \dot{\beta}^*_G[1]\cdot H .
\end{equation}
The Peierls formula for the Poisson bracket of observables of the form
in Equation~\eqref{eq:linobsv-equiv} was considered explicitly
in~\cite[Sec.4.4]{kh-peierls} (see also~\cite{fs} and~\cite[Ex.3.8]{hs})
and is given by the formula
\begin{equation}\label{eq:peierls}
	\{A,B\}[G]
	= \{\dot{A}_G,\dot{B}_G\}_G
	= \int_{M\times M} \dot{\alpha}^*_G[1](x)\cdot E_G(x,y) \cdot \dot{\beta}^*_G[1](y)
\end{equation}
where $E_G(x,y) = E^+_G(x,y) - E^-_G(x,y)$, with $E^\pm_G(x,y)$ being the
integral kernels of the retarded and advanced Green functions of the
so-called \emph{Lichnerowicz operator} (which is a hyperbolic
differential operator obtained from a de~Donder gauge fixing of the
linearized Einstein equations) of the background metric $G$.

The result is gauge invariant, that is $\{A,B\}[\chi^*G] = \{A,B\}[G]$
for a diffe\-o\-mor\-phism $\chi^*\colon M\to M$, essentially by construction.
More explicitly, since each of the elements in the formula is
invariantly constructed from the metric $G$, the following identities
hold: $\dot{\alpha}^*_{\chi^*G} = \chi^* \dot{\alpha}^*_{G}$,
$\dot{\beta}^*_{\chi^*G} = \chi^* \dot{\beta}^*_{G}$ and
$E_{\chi^*G}(x,y) = (\chi,\chi)^* E_G(x,y)$, where $(\chi,\chi)^*\colon
M\times M \to M\times M$ is defined in the obvious way. Combining these
identities with formula~\eqref{eq:peierls} explicitly shows that
$\{A,B\}$ is a gauge invariant (though now distributional bilocal,
instead of local) observable.

It is also worth examining whether the linearized observable
$\dot{A}_G[H]$ fits the criteria of being a gauge invariant observable
for linearized gravity on the background $G$. The answer is of course
Yes, as follows from the identity $\Lie_v \alpha[G] =
\dot{\alpha}[\Lie_v G]$, where $\Lie_v$ is the Lie derivative with
respect to a vector field $v$, which is the linearized version of the
invariance property $\chi^*\alpha[G] = \alpha[\chi^*G]$, and the Cartan
magic formula $\Lie_v \alpha[G] = \d\left( \iota_v \alpha[G] \right)$
for top-degree forms. For convenience, let us also define the
differential operator $K_G[v] = \Lie_v G$, which we will call the
\emph{Killing operator}. The gauge invariance condition for
$\dot{A}_G[H]$ in linearized gravity consists in the requirement that
$\dot{A}_G[K_G[v]] = 0$ for any vector field $v$. This follows from the
preceding identities:
\begin{equation}
	\dot{A}_G[K_G[v]]
	= \int_M \dot{\alpha}_G[\Lie_v G]
	= \int_M \Lie_v \alpha[G]
	= \int_M \d\left( \iota_v \alpha[G] \right) = 0 ,
\end{equation}
where the last equality follows from the fact that $\iota_v \alpha[G]$
has compact support by the locality hypothesis. Thus, $\dot{A}_G[H]$ is
a linear, local and gauge invariant observable in linearized gravity.

Let us recall the notion of linear, local and gauge invariant observable
from~\cite{fs} (also~\cite[Sec.4.4]{kh-peierls}, \cite[Ex.3.8]{hs}),
which is an observable of the form
\begin{equation}\label{eq:lingrav-obsv}
	C[H] = \int_M \gamma\cdot H ,
\end{equation}
with a compactly supported section $\gamma\colon M \to \tilde{F}^*$ that
satisfies the condition $K_G^*[\gamma] = 0$, where $K^*_G$ is the formal
adjoint of the Killing operator $K_G$. More explicitly, there exists a
form-valued bidifferential operator $W_K$ such that $\gamma\cdot K_G[v]
- K^*_G[\gamma]\cdot v = \d W_K[\gamma,v]$ for any vector field $v$ and
any section $\gamma\colon M\to \tilde{F}^*$; $K_G^*$ is equivalent to
the divergence of a symmetric $2$-tensor.
\begin{prop}
Given the linearized observable $\dot{A}_G[H]$, as discussed above,
there always exists a local observable $C[H]$ in linearized gravity of
the form~\eqref{eq:lingrav-obsv} such that $\dot{A}_G[H] = C[H]$.
\end{prop}
\begin{proof}
For this result to hold, it is clearly sufficient that there exist a
compactly supported section $\gamma\colon M\to \tilde{F}^*$, satisfying
$K_G^*[\gamma] = 0$, and a form-valued linear differential operator
$\mu[H]$, with compact support for arbitrary argument $H\colon M\to F$,
such that $\dot{\alpha}_G[H] = \gamma\cdot H + \d\mu[H]$. We shall
construct such $\gamma$ and $\mu[H]$ explicitly.

Recall the identity $\dot{\alpha}_G[H] = \dot{\alpha}^*_G[1]\cdot H +
\d W_\alpha[1,H]$. We set $\mu[H] = W_\alpha[1,H]$ and $\gamma =
\dot{\alpha}^*_G[1]$. It remains to show that $K_G^*[\dot{\alpha}_G^*[1]]
= 0$. Note that, from the gauge invariance of $\dot{A}_G[H]$ discussed
earlier, we already know that
\begin{align}
	\dot{\alpha}_G^*[1]\cdot K_G[v]
\notag
	&= \dot{\alpha}_G[K_G[v]] + \d W_\alpha[1,K_G[v]] \\
\label{eq:decomp1}
	&= \d\left( \iota_v \alpha[G] + W_\alpha[1,K_G[v]] \right)
\end{align}
for an arbitrary vector field $v$. On the other hand, we also have the
equality
\begin{equation}
\label{eq:decomp2}
	\dot{\alpha}^*_G[1]\cdot K_G[v] = K_G^*[\dot{\alpha}_G^*[1]]\cdot v
		+ \d W_K[\dot{\alpha}^*_G[1],v] .
\end{equation}
The final tool that we need to invoke is the well-known
fact~\cite[Thm.4.7]{olver} that, for any top-degree form valued linear
differential operator $\psi[v]$, in any decomposition of the form
$\psi[v] = \phi\cdot v + \d\xi[v]$ the coefficients $\phi$ and the term
$\d\xi[v]$ are unique (in particular $\phi=\delta_{EL}[\psi[v]]$ is the
Euler-Lagrange derivative of $\psi[v]$). Thus, comparing
Equations~\eqref{eq:decomp1} and~\eqref{eq:decomp2}, we find that
$K_G^*[\dot{\alpha}^*_G[1]] = 0$, as was desired.
\end{proof}

\section{Discussion}\label{sec:discuss}

In this note, we have discussed the notion of local observables in field
theory, advocating that the standard notion of locality
(Section~\ref{sec:locobsv}) should be relaxed in a well-defined way
(Section~\ref{sec:glocobsv}). We have argued that the two motivating
properties of local observables, diffusion of UV singularities and IR
regularization, still hold for generalized local observables in the
sense defined in Section~\ref{sec:glocobsv}.

A small price to pay is that a generalized local observable may be naturally
defined as functions only on an open subset%
	% XXX: Explicit footnote number used only because LaTeX complains
	% with "counter too large" error otherwise.
	\footnote[1]{A related mathematical phenomenon occurs in complex and
	algebraic geometry. Certain complex and algebraic varieties have very
	few globally defined functions. By restricting to open subsets, many
	more functions can be considered, that otherwise developed
	singularities if extended to the entire space. Such partially defined
	functions are studied in the theory of \emph{sheaves}. We have not
	developed this analogy in detail because there is not yet a clear
	application of sheaf theory in this context, other than as a concise
	terminology.} %
of the full phase space of the field theory. Classically, it is no
problem to restrict one's attention to an open subset of the full phase
space. If needed, such an observable may be extended to the full phase
space by appealing to basic results in differential topology. We have
shown that linearization about a specific point of the configuration
space gives a gauge invariant observable for linearized gravity on the
corresponding background, irrespective of how large is the neighborhood
of the linearization point on which the observable can be defined. That
is of course the expected result for the linearization of an observable
invariant under full non-linear gauge transformations. We expect the
same behavior at any order of perturbation theory; the truncated
expansion of the observable should be invariant under perturbative gauge
transformations truncated at the same order, which is sufficient for the
purposes of perturbative quantization.

It is well-known that gravitational theories do not admit any
non-trivial local observables that are also gauge-invariant. Hence, it
is a significant advantage of the new definition that the class of
generalized local observables in gravitational theories does admit a
large number of observables that are gauge invariant
(Section~\ref{sec:diffinv}). We have given a typical example of one such
observable, motivated by an old proposal of Komar and
Bergmann~\cite{bg0,bg1}. In fact, such gauge invariant observables are
sufficient to separate the gauge orbits on a large open subset of the
phase space (Theorems~\ref{thm:generic} and~\ref{thm:sep}). The main
technical tool in the construction of these gauge invariant observables
is the theory of differential invariants, which in the literature on
GR are also known as curvature scalars or curvature invariants.

Unfortunately, the large open subset of the phase space mentioned above
specifically excludes solutions that have a high degree of symmetry.
Some of these symmetric solutions can be of great physical importance,
at least in GR, with examples like Minkowski or
Schwarzschild or de~Sitter spacetimes. The reason for the exclusion is
that observables based on curvature scalars are incapable of separating
certain inequivalent gauge equivalence classes of solutions. At the
geometric level, the same phenomenon manifests itself in the fact that
the moduli space of Lorentzian metrics (the quotient of jets of
Lorentzian metrics by the action of diffeomorphisms) is not
Hausdorff~\cite{gns}. A well-known example is that all curvature scalars
vanish both on flat Minkowski spacetime as well on non-flat null pp-wave
spacetimes (non-linear wave gravitational wave
solutions)~\cite{hc,chp,chp2}. This is problematic if one would like to
connect perturbative theory about Minkowski space with non-linear local
observables of the kind discussed above. In principle, it is known that
there exist non-scalar differential invariants that are capable of
locally distinguishing non-isometric Lorentzian metrics (cf.\ the
Cartan-Karlhede algorithm discussed in~\cite[Ch.9]{stephani-sols} and
references therein). At this point it remains an open problem to be
investigated whether these more refined differential invariants could be
used to construct local (or perhaps multilocal) observables that are
capable of separating all gauge orbits on the phase space of GR and
other gravitational theories.

In Section~\ref{sec:pois}, we showed that generalized local and gauge
invariant observables have gauge invariant Poisson brackets using the
Peierls formula. However, Poisson brackets of local observables are in
general no longer local. At best they could be described as multilocal
with distributional smearings. Such observables have been previously
discussed in the literature~\cite{bf-notes,bfr}, with careful attention
paid to the class of distributions that can be consistently allowed to
construct an algebra of multilocal observables closed under Poisson
brackets. The added complication in gravitational theories, as is
evident from the Peierls formula, is that in order to preserve gauge
invariance we must allow distributional smearings themselves to depend
on the metric and possibly other dynamical fields. Thus, another
important avenue for investigation is the generalization of multilocal
observables to allow for field-dependent distributional smearings.

It might be argued that the local and gauge invariant observables that
we have introduced in this note are of a \emph{relational} kind
(see~\cite{tamb} and references therein). However, they do not
automatically come with a phenomenological interpretation. That is,
given a particular observable of this kind, it may not be immediately
clear what kind of experimental protocol would be modeled by it (this
issue is discussed clearly in~\cite{kh-time1}). On the other hand, there
is some existing literature that has considered relational observables
in linearized and perturbative gravity with more clear phenomenological
interpretations, but ran into UV divergences in explicit
computations~\cite{ford-lightcone, ford-top, ford-focus, ford-angle,
woodard-thesis, tsamis-woodard, ohlmeyer, kh-time1, kh-time2}. Perhaps
replacing the overly singular proposed observables in these references
with regularized versions written as local and gauge invariant
observables would yield a double benefit: provide certain local
observables with phenomenological interpretations, diffuse UV
singularities in explicit computations. As a further step, it would be
most interesting to identify local gauge invariant observables that
would model some aspects of the data collected by cosmological
observations, such as the Cosmic Microwave Background temperature
fluctuations and its polarization.

It should also be mentioned that another attempt~\cite{gmh} to write
down relational observables (though without clear phenomenological
interpretations) using curvature scalars ran into IR divergences in
explicit computations. On the other hand, our local observables are
designed to be IR regularizing and might give better results in similar
computations.

\ack
The author would like to thank Valentin Lychagin for helpful discussions
on the topic of differential invariants. Thanks also to Jochen Zahn for
comments on an earlier version of the manuscript. The kind hospitality
of the Department of Mathematics at the University of York, where part
of the manuscript was completed, is also acknowledged.

%\appendix

\section*{References}
\bibliographystyle{utphys-alpha}
\bibliography{paper-grobsv}

\end{document}